\documentclass{article}

\usepackage[dvips,dvipdfm,ps2pdf,hyperfigures,colorlinks]{hyperref}
\usepackage{amsmath,amssymb,graphicx}
\usepackage{fullpage}

\newtheorem{theorem}{Theorem}[section]
\newtheorem{lemma}[theorem]{Lemma}
\newtheorem{corollary}[theorem]{Corollary}

\newtheorem{definition}[theorem]{Definition}

\newcount\proofqeded
\newcount\proofended
\def\qed{ \mbox{\ \vrule width1ex height1em depth0cm}
\global\advance\proofqeded by 1 }
\newenvironment{proof}{\proofstart}{\ifnum\proofqeded=\proofended\qed\fi \global\advance\proofended by 1
  \medskip}
\makeatletter
\def\proofstart{\@ifnextchar[{\@oprf}{\@nprf}}
\def\@oprf[#1]{\paragraph{Proof #1:~}}
\def\@nprf{\paragraph{Proof:~}}
\makeatother

\newcommand{\Exp}{\operatornamewithlimits{\mathbb{E}}}

\newcommand{\pcube}{\operatorname{\{\pm1\}}}
\newcommand{\eps}{\epsilon}

\newcommand{\sgn}{\operatorname{\text{{\rm sgn}}}}

\newcommand{\calf}{{\cal F}}

\newcommand{\Fix}{{\sf Fix}}
\newcommand{\defeq}{\stackrel{{\mbox{\tiny def}}}{=}}
\newcommand{\re}{\mathbb{R}}

\newcommand{\ignore}[1]{}

\newcommand{\comment}[1]{\emph{#1}}
\renewcommand{\comment}[1]{}

\renewcommand{\paragraph}[1]{\noindent {\bf #1}}
\begin{document}

\title{Robust Regulatory Networks}

\author{Arnab Bhattacharyya\footnote{Authors are in alphabetical order and contributed equally to the paper} \footnote{\texttt{abhatt@mit.edu}.  Research supported in part by a DOE Computational Science Graduate Fellowship and NSF Grants 0514771, 0728645 and 0732334.}  \and Bernhard Haeupler \footnotemark[1] \footnote{\texttt{haeupler@mit.edu}.  Research supported by an Akamai MIT Presidential Fellowship}}
\date{}

\maketitle
\vspace{-1em}
\begin{center}
{ Massachusetts Institute of Technology\\ Computer Science and Artificial Intelligence Laboratory\\ Stata Center, 32 Vassar Street, Cambridge, MA 02139} 
\end{center}

\vspace{3em}

\begin{abstract}
One of the characteristic features of genetic networks is their inherent robustness, that is, their ability to retain functionality in spite of the introduction of random errors.  In this paper, we seek to better understand how robustness is achieved and what functionalities can be maintained robustly.  Our goal is to formalize some of the language used in biological discussions in a reasonable mathematical framework, where questions can be answered in a rigorous fashion.  These results provide basic conceptual understanding of robust regulatory networks that should be valuable independent of the details of the formalism.

We model the gene regulatory network as a boolean network, a general and well-established model introduced by Stuart Kauffman.  A boolean network is said to be in a viable configuration if the node states of the network at its fixpoint satisfy some pre-specified constraint.  We define how mutations affect the behavior of the boolean network, and we say a network is robust if most random mutations of the model reach a viable configuration.

We first study the case when the boolean network is specified by a directed acyclic graph.  Random mutations induce a neighborhood around the configuration that would be reached if there were no mutations.  We show that for the case of acyclic networks, this neighborhood is a bijective transformation of the usual Hamming neighborhood.  A robust acyclic network chooses the bijection so that most of the neighborhood lies in the space of viable configurations.  The greater the degree of the network, the more complex the bijection is allowed to be and thus the greater the possibility of robustly satisfying constraints.

Next, we study networks where directed cycles are present.  We show that cyclic networks can make the volume of the neighborhood smaller, by mapping different errors in the network to the same final configuration.  Thus, cyclic networks can be dramatically more powerful with respect to robustness than acyclic networks.  Also, we explicitly describe a large class of constraints for which cyclic networks provide robustness.
\end{abstract}


\mbox{}\\

\section{Introduction: Understanding Robustness}

One of the hallmark features of life is its diversity.  To solve the same basic problems of survival and reproduction, nature has devised a scintillating array of solutions, each remarkably different from others in many aspects.  In the long term, such biological innovation is necessary, because changing environmental conditions mean that some solutions will become defunct while others will become more advantageous. The fundamental question that arises then is how does innovation in
biology arise?

To this end, Wagner in \cite{wag05b} defined a biological system to be {\em evolvable} ``if it can acquire novel functions through genetic change, functions that help the organism survive and reproduce.''  What exactly are the features of a system that displays evolvability remain a mystery.  The question of evolvability can also be phrased in terms of the genotype-phenotype map.  The genotype of an organism is the hereditary information contained in the genome, while the phenotype is the set of properties actually exhibited by the organism and acted upon by natural selection.  For evolvability, the map between the genotype and phenotype must be such that random mutations to the genotype can possibly ``improve'' the phenotype so that novel functionality is acquired.  In nature, this map is not the identity map: there is a complex translation process that constrains the phenotypes that can be expressed while also encouraging variability.  One can ask then for the properties of the genotype-phenotype map for evolvable biological systems.

In this paper, we rigorously study one specific aspect of evolvability:
{\em robustness}.  At a first glance, robustness and evolvability seem to
be diametrically opposed concepts.  Robustness refers to a system's ability to
retain functionality in the presence of changes, while evolvability
refers to a system's ability to acquire new functionality.  To
resolve this dilemma, let us define robustness more precisely,
following \cite{wag05b}.  Given a specific phenotypic property $f$, we
say a mutation to the genotype is {\em neutral with respect to $f$}
if that mutation does not affect possession of the property $f$.  We say a system is
{\em robust with respect to $f$} if the vast majority of mutations are
neutral with respect to $f$. For instance, the phenotypic property to
be preserved could be the RNA secondary structure, which a prerequisite for RNA function.  Then genetic change in an RNA
molecule that is neutral with respect to RNA secondary structure would
preserve the RNA's secondary structure but potentially change other
aspects.  Another example (also discussed in \cite{wag05b}) is cryptic
variation in developmental genes.  These mutations preserve the
development of complex organs, such as the eye and legs, under usual
circumstances but could drastically alter their development in alternate environments \cite{ruthlind}.  In this
case, the property being preserved is development of these organs
in a specific environmental and genetic background.

Now, we can explain  why robustness with respect to a
given phenotypic property increases evolvability.  If a system is
robust with respect to some primary property $f$, then, since most
mutations are neutral with respect to $f$, the system can express many
phenotypes satisfying $f$ and thus has a higher chance of encountering
a phenotype that satisfies some other property $g$.  Thus, novel
phenotypes can be discovered while not destroying previous
functionalities already achieved.  Gould called this process {\em
  exaptation} (\cite{gould}) to refer to organismal features that
become adaptations to new conditions long after they have already
arisen to achieve some more basic functionality.  In other words,
robustness allows a system to accumulate a reservoir of neutral
mutations and thus has a greater potential for innovation with respect to
new unexplored functionalities. 

Evidence for robustness with respect to particular phenotypic properties is abundant
throughout nature.  Although we describe robustness as
resilience to mutations in the genome, similar notions also exist for
changes at all levels of organization.  Proteins tolerate thousands of
amino acid changes, metabolic networks continue to function after
removal of intermediate steps, gene regulatory networks are unaffected
by alteration of gene interactions, genetic changes in embryonic
development often hardly affect the viability of the adult organism,
and microbes and higher organisms can tolerate complete elimination of
many genes.  Organization in biological structures is incredibly
complex, and it is often a matter of great mystery how such robustness
can be achieved at all.

This paper is part of an effort to better understand the robustness
property and the environments under which robust solutions are
possible.  Our goal is to formalize some of the language used in
biological discussions in a reasonable mathematical framework, where
questions can be answered in a rigorous fashion. 

To simplify the discussion, let us only look at the case of robustness
of the phenotype to mutations in the genotype (although much of our
results can potentially be applied to robustness at other levels of
organization).  The genotype takes the form of a {\em regulatory gene
network}.  The expression level of each gene is functionally related
to the expression level of some other genes.  Thus, if the expression
level of one gene is changed, expression levels of other genes are
also modified according to the regulatory connections. To model these
networks we quantize gene expressions by two levels, {\sc ON} and {\sc
  OFF} and describe their interaction via boolean networks, introduced
by Kauffman \cite{kf69,kf93,kf95}. This is a very general and
well-established model.  The phenotype expressed by a particular
genotype (represented by the regulatory network) is described by the
stable configurations of the network; this correspondence has has been verified using
simulation and evaluating gene expression data  \cite{AlbertOthmer,MendozaAlvarez,Espinosa-Soto,HuangIngber,Huang,Gardner}.

In this paper we take this widely used biological model and try to
identify, describe and understand conditions under which robust gene
networks can work -- i.e. which conditions can be robustly satisfied by
regulatory networks. We furthermore want to investigate how
such robust solutions can be found -- an everyday task of nature. We
believe that answers to these questions are essential for
understanding the biology of evolution, the power and structure of
regulatory (gene) networks and their interaction with mutations.

\paragraph{\textbf{Organization.  }}  We continue in Section \ref{sec:defs} by
giving formal definitions for our model and state our main results,
while relating them back to the biological motivation discussed above.
In Section \ref{sec:results}, we give a more detailed discussion of our results
along with proofs.  Finally, we end with some conclusions and
suggestions for future work in this area.

\section{The Model: Networks and Mutations\label{sec:defs}}

\subsection{Definitions}

In this section we give the formal definition of the model we will
investigate; we biologically justify the details of our formalism in Section \ref{sec:defmot}. 

The phenotype is specified by gene expression levels, modeled as $n$ boolean\footnote{A boolean domain has exactly two values with interpretations as \textit{True} and \textit{False}. We use these interpretations with the functions $\vee$ (OR), $\wedge$ (AND) and $\bar{\ }$ (NOT) interchangeably with the values $1$ (\textit{True}) and $-1$ (\textit{False}).} characters
$x_1,\dots,x_n \in \pcube$. The phenotype $(x_1,\dots,x_n)$ is said to be {\em viable} exactly
when $f(x_1,\dots,x_n) = 1$ where $f : \pcube^n \to \pcube$ is the
constraint determined by the environment or genetic background to be
satisfied.   We call $f$ the {\em objective function}. 

Next, we formalize the definition of the genotype  by
asserting that the genotype encodes a boolean network with the gene
expressions as boolean variables. Such a {\em boolean network
  $N=(x,u_1,\dots,u_n)$} is specified by $n$ boolean variables $x = (x_1,\dots,
x_n)  \in \pcube^n$ and $n$ corresponding {\em update functions}
$u_1,\dots, u_n: \pcube^n \to \pcube$ describing how the variable $x_i$ depends on
all other variables.

Every network $N=(x,u_1,\dots,u_n)$ induces a directed graph
$G_N=(V_N,E_N)$ on its variables $V_N = \{x_1,\ldots,x_n\}$ where 
$(x_i,x_j)$ is an edge in $E_N$ iff the variable $x_i$ has an influence on the
update function $u_j$ of node $x_j$ (i.e., $\exists
x_1, \ldots, x_{j-1}, x_{j+1}$\\$, \ldots, x_n$ such that 
$u_j(x_1,\ldots,x_j=-1,\ldots,x_n) \neq
u_j(x_1,\ldots,x_j=1,\ldots,x_n)$) The graph $G_N$ describes how
changes can in principle ``spread'' through the network $N$.

Next we want to describe how exactly the states of the boolean
variables of a network
change over time. For this we define the {\em configuration} of a
boolean network as an assignment $\alpha \in \pcube^n$ to all its
variables (i.e. $\forall i=1,\ldots,n \ :\  x_i = \alpha_i$). Then we
define a dynamic system for every boolean network $N=(x,u_1,\dots,u_n)$ together
with an initial configuration $\alpha$ by inductively defining a
sequence of configurations: 
\begin{align*}
x(1)   &\ = \ \  \alpha \\
x(t+1) &\ = \ \ (u_1(x(t)),\ldots,u_n(x(t))) & \forall t \geq 1
\end{align*}
This dynamic system gives a sequence of configurations for every time
$t$. For time $t=1$ this is the initial configuration $x(1)=\alpha$
and for later times the configuration $x(t)$ is formed by applying the
update functions on the last configuration.  Note that the states of
all the nodes are updated synchronously.

Since the configuration space is finite and the dynamics of the
network are deterministic, the network will eventually fall into a
previously visited  configuration, after which the configuration
dynamics become periodic.  This cyclic trajectory is called an {\em
  attractor}.  If the attractor is just one configuration -- a cycle
of length 1 -- it is called a {\em fixpoint}.  It is clear that given
an initial configuration $\alpha$, if a network does reach a fixpoint
starting from $\alpha$, it is unique.  Assuming the network $N$ does reach
a fixpoint starting from $\alpha$, we denote it by $\Fix(N,\alpha)$;
if it does not, $\Fix(N,\alpha)$ is undefined.

Next, we specify how mutations change the boolean network. As
motivated below in Section \ref{sec:defmot}, we look at mutations
modifying the update functions that do not change the topology $G_N$
of the network.  We say that a network $N'$ is a {\em mutation of a network}
$N$ if $N = (x,u_1,\dots,u_n)$ and $N' = (x, v_1,\dots,v_n)$ with
each $v_i$ either equal to $u_i$ or $-u_i$.   Note that inverting an update function does not
change whether one variable has an influence on an other variable, giving
always $G_N = G_{N'}$.  Figure \ref{fig:network} shows an example of a network and a
mutation of it.  Given a  {\em mutation parameter}   $\eps \in
(0,\frac{1}{2})$, we define an {\em 
  $\eps$-mutation} of a boolean network $N=(x,u_1,\dots,u_n)$ to be a random variable
denoting a boolean network $N'=(x,v_1,\dots,v_n)$ with the same variables but
changed update functions. Specifically, independently for every $i \in
[n]$, the new update function $v_i$ is equal to its original $u_i$ with probability
$1-\eps$ and its complement $v_i = -{u}_i$ with probability
$\eps$.

\begin{figure}
\begin{center}
\includegraphics[width=15cm]{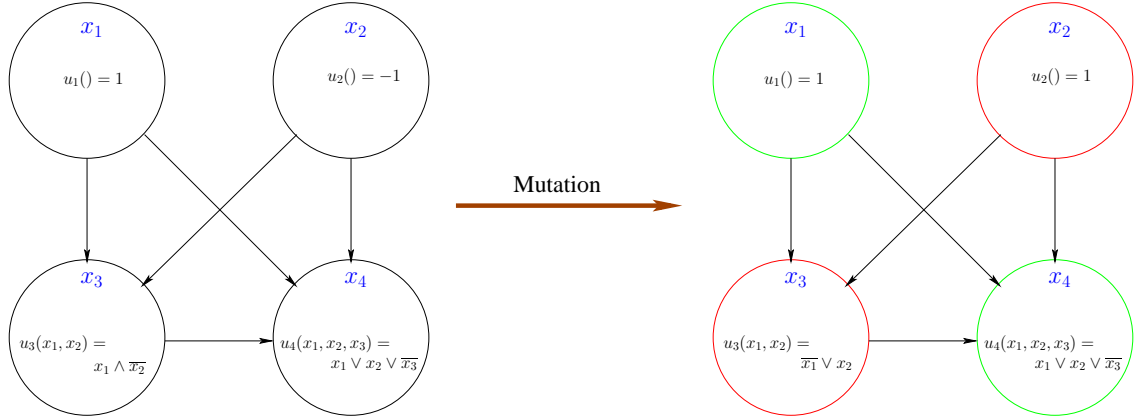}
\end{center}
\caption{{The boolean network
    on the left reaches the final configuration $(x_1,x_2,x_3,x_4) =
    (1,-1,1,1)$.  The mutated    network on the right reaches the
    final configuration $(x_1,x_2,x_3,x_4) = (1,1,1,1)$.}}
\label{fig:network}
\vspace{.5cm}
\end{figure}

Finally, we want to quantify how well a boolean network satisfies a
given objective function in the presence of mutations. Formally, for a
mutation parameter $\eps \in (0,\frac{1}{2})$ and a survival probability
$\delta \in [0,1]$ we say that a network $N$ with an initial
configuration $\alpha \in \pcube^n$ is {\em $(\eps,\delta)$-robust
  with respect to} an objective function $f : \pcube^n \to \pcube$
iff the probability that the fixpoint $x(\infty) = \Fix(N',\alpha)$
reached by an $\eps$-mutation $N'$ of $N$ satisfies $f(x(\infty)) =
1$ is at least $\delta$.  Note that this definition only makes sense
if the network $N'$ reaches a fixpoint starting from $\alpha$; see Section \ref{sec:whyfix}
for biological motivation.  We say
that the network $N$ is {\em optimally $\eps$-robust with respect to}
$f$ if it achieves the largest survival probability $\delta$ among all
networks that are $(\eps,\delta)$-robust with respect to $f$. \\ 
The central question we address in this paper is whether given an
objective function $f$, there exists a network $N$ that is $(\eps,
\delta)$-robust with respect to $f$ with $\eps$ a constant and
$\delta$ very close to $1$.  In order to formalize ``very close'', we
parametrize the objective function by $n$, the number of boolean
characters, and then desire that $\delta$ goes to $1$ asymptotically
as $n$ approaches infinity (holding $\eps$ constant).
For a mutation rate $\eps \in (0,\frac{1}{2})$, a family of objective
functions $\{f_n : \pcube^n \to  \pcube\}_{n=1,\dots}$ is said
to be {\em $\eps$-robustly expressible}  iff, for every $n$, there
exists a boolean network $N_n$ on $n$ variables and a configuration
$\alpha_n \in \pcube^n$ such that $N_n$ with initial configuration
$\alpha_n$ is $(\eps,\delta_n)$-robust with respect to $f_n$, where the
survival probabilities $\delta_n$ go to $1$ when $n$ approaches infinity. A family of
objective functions is is said to be {\em robustly expressible} if it is
$\eps$-robustly expressible for some constant $\eps \in
(0,\frac{1}{2})$.  We will sometimes abuse notation by saying that an objective
function (instead of a family of objective functions) is robustly
expressible; here, the function is implicitly parametrized by its arity.

\subsection{The Greater Picture\label{sec:defmot}} 
Let us reconnect the formal definitions above in Section \ref{sec:defs} to the biological
motivation discussed in the introduction.  The genotype is described
by a boolean network, which mimics the dynamics of the genetic
regulatory network.  Boolean networks are a well-established model to
describe regulatory network interactions (see, for instance, the study
of floral organ development using such a model in
\cite{Espinosa-Soto}).  These network models are often devised by
discretizing nonlinear continuous models but, for our purposes, this
is inessential.  For a given genotype represented as a boolean
network, the phenotype is assumed to be the gene expression levels, given
by the stable configuration of the network.  As mentioned earlier, the 
correspondence between the gene expression levels and the stable configuration
of the network has been shown to be valid in several biological experiments (see e.g.
\cite{AlbertOthmer,MendozaAlvarez,Espinosa-Soto,HuangIngber,Huang,Gardner}).  
Nevertheless one should note that modelling the phenotype by the gene expression levels 
is a gross simplification of reality.  For example, the environment and other epigenetic
factors also often play an important role.  But as a first attempt at a systematic 
understanding of the robustness of regulatory networks, our model should suffice.

The particular phenotypic property with respect to which robustness is
ascertained is captured in the boolean objective function.  Note that
the objective function could be highly complicated.  For example, if
the objective function determines whether a piece of RNA forms a particular secondary structure, 
it would have to encode a procedure for determining the
secondary structure from a given RNA primary structure (a problem for which no
efficient algorithms are known in general).  As environmental and
background genetic conditions can be very complex, it is difficult to
say anything specific about the structure of the objective functions
that arise in nature.

In our analysis of robustness, we hold the mutation rate constant and let the number of characters in the
phenotype grow arbitrarily large, since in reality the number of genes is very large numerically and, in many
models, the mutation rate is independent of the number of genes.  
We discuss some further modelling issues below.\\
\subsubsection{Degree of Regulatory Networks.}  In nature, regulatory
interactions between genes is implemented 
through the presence of regulatory regions in the genome.  Regulatory
proteins, such as promoters and inhibitors, bind to this region and
either encourage or discourage expression of the associated gene.  The
regulatory region is a short stretch of DNA and so it is not feasible
to have too many regulatory proteins binding to a single regulatory
region simultaneously.  In our model, this means that the in-degree (the number of adjacent
incoming edges) of nodes in the boolean network representing the genotype should be
small.  For a boolean network $N$, we term the maximum in-degree of a
node in $G_N$ the {\em degree} of the network $N$.  When the degree is $0$, then the
network essentially consists of a static configuration, which does not
seem too interesting from a biological perspective.  And when the
degree is $n$, a node in the network can potentially interact with
every other node in the network, an impractical situation.  We will be
interested then in the tradeoff between degree and robustness of
boolean networks.

\subsubsection{Attractors of Boolean Networks.\label{sec:whyfix}}  As noted above, we
associate phenotypes with fixpoints of boolean 
networks. But clearly, boolean networks can also reach a cycle instead
of a fixpoint as an attractor.  In this paper, we do not consider
such networks to robustly express any objective function.  The primary
reason for this restriction is that in most modelling of real
biological systems by boolean networks, the attractors reached are
found to be fixpoints rather than cycles \cite{Espinosa-Soto}.  
There are instances in the literature, such as \cite{Aldana}, where
the phenotypes are taken to be attractors, whether they be
fixpoints or cycles.  But it is not clear if such a setup is indeed
reasonable from the perspective of modelling biological systems.

\subsubsection{Mutation Models.}  The literature shows many approaches for
modelling how mutations act on the genotype, or the boolean network in our model.  Mutations could
modify the network by changing the adjacency relations (as in 
\cite{BraunewellBornholdt}), changing the regulatory interactions (as
in \cite{SzejkaDrossel}), or duplicating and deleting nodes (as in
\cite{Aldana}). For this paper we want to investigate the case when
a mutation on the genotype modifies the regulatory interactions
between nodes by changing selected update functions to their
complements.  We believe that such errors capture the most common types
of mutations in nature. These mutations to the boolean network can
either be genetic changes passed down from one generation to the next,
or be due to environmental disturbances occurring in a single individual.

It is important to notice that in our setting, some of the other
choices of mutation models do not give interesting results.  For
instance, if 
errors are only on edges, then a robust network would be a static
assignment, i.e., a network with no edges, which is not so
interesting.  Also, if mutations change the update functions
arbitrarily, then the graph induced by the network changes arbitrarily
which would make robustness with respect to non symmetric objective
functions impossible to achieve.  Additionally, it is not very
biologically motivated to consider a gene changing its entire set of
regulators after one mutation.
Thus, for a mutation model to be mathematically interesting and
biologically relevant, we require that the network robustness should
not decrease if the network degree is increased and that if $N'$ is a
mutation of $N$, $G_{N'}$ should be somehow very closely related to
$G_N$.  In this paper, our mutation model obeys $G_N = G_{N'}$.

For other mutational models which do give interesting results, such as
errors on both edges and nodes, we believe that many of the results
presented here apply in spirit to those models as well and, also, that
many of the techniques we describe for constructing robust networks
have analogues in the alternate models.

\subsubsection{Constructivity of Robust Networks\label{sec:algo}}  Robust
expressibility of an objective function just asserts that there 
exists a network robustly expressing it.  However, for nature to be
able to use the robustness property, it needs to be able to construct
the network in some fashion.  Perhaps it starts with a non-robust
network and makes small changes to it to drive it toward robustness,
in an evolutionary process, or perhaps it uses some other procedure,
but a basic requirement for the robust expressibility property to be
useful is that the robust network be {\em efficiently constructible}.
More precisely, for an objective function $f : \pcube^n \to \pcube $,
we say that a network $N$ robust with respect to $f$ is efficiently
constructible if there is a polynomial time algorithm that has access
to $f$ as an oracle and that outputs a description of $N$.   For an
arbitrary objective function, one cannot hope for efficient
constructibility.  However, one would like to show efficient
constructibility for objective functions belonging to special families
which have small description size.  Also, note that efficient
constructibility is a weak condition to impose, since nature might be
restricted to a weak computational model.

\subsection{Previous Work}

The study of evolvability and the origin of novelty in biological
systems is an intensively studied area in biology. 
\textit{The Plausibility of Life} (\cite{kg06}) by Kirschner and Gerhart is a great introduction to the
current understanding in this area, containing pointers to many
relevant articles in the field.  The
connection between robustness and evolvability is described well in
\cite{wag05b} and in \cite{ciliberti}.

Boolean networks were originally introduced by Kauffman in \cite{kf69}.
In \textit{The Origins of Order} (\cite{kf93}), Kauffman explains his position that biological systems display
the properties of an ensemble of random boolean networks with
parameters that make them lie at the `edge of chaos', that is, near a
statistical phase transition.   The concept of designing boolean
networks to solve specific problems is a newer idea,
discussed by Hasty et. al. in \cite{hasty} for example. Genetic
regulatory networks have been modelled by boolean networks in
many parts of the literature now (\cite{AlbertOthmer,MendozaAlvarez,Espinosa-Soto,HuangIngber,Huang,Gardner}).

Robustness of boolean networks against various types of faults have
also been under investigation.  \cite{Aldana,SevimRikvold} study the
resilience of 
random boolean networks to mutations in the update functions, while
\cite{BraunewellBornholdt} and \cite{SzejkaDrossel} study resilience
of single boolean networks to other types of faults.  These papers
offer numerical/experimental evidence in favor of robustness and do
not study how robustness arises formally.  In a non-biological
context, robustness was also studied by Hornby in \cite{hornby} in the
context of developing programs expressing the design of furnitures!  He was
interested in representations of tables that one could mutate and
still retain a table generating program (but with other potentially
useful features).  Finally, evolvability of evolutionary programs was
studied by Reisinger and Miikulainen in \cite{reis}, but  in an
empirical fashion.  

As far as we know, this is the first rigorous
investigation of robustness and evolvability.  As mentioned above,
there has been a stream of previous papers examining theoretical
models for the regulatory system and then showing through simulation
how robustness can arise.  What we believe is novel in our current
work is that we obtain a mathematical understanding of how the level of
robustness is related to the structure of the regulatory network
model.  There have been
previous papers which have suggested the need for a formal
study, such as \cite{evolevol} and \cite{bigsurvey}.  
Furthermore, the philosophical underpinnings of our study are
slightly different from most earlier work.  Most of the previous papers
model the genotype to be uniformly generated by a process that is
described by a few parameters.  For instance, in the standard Kauffman
model, random boolean networks ({\em $NK$ networks}) are uniformly
generated by a probabilistic process that takes as inputs a parameter
$N$ for the number of nodes and a parameter $K$ for the degree of
each node. On the other
hand, the genotype model in our work is far richer.  We allow the
genotype structure to be as complicated as desired and then optimize for
robustness.  The difference between the two approaches is
characteristic of the difference between the methodologies of physics
and computer science.  Carlson and Doyle
\cite{carlsondoyle1,carlsondoyle2,carlsondoyle3} have introduced a 
conceptual framework, known as {\em Highly Optimized Tolerance}, that argues that biological systems are highly structured and
optimized for robustness and that they must be described by a large number
of parameters; our work can be loosely viewed as fitting into this
framework.

\subsection{Our Results}

We start our investigation of robustness by focusing on the special
class of {\em acyclic boolean networks}.  These networks have the
feature that they are guaranteed to reach fixpoints
starting from any initial configuration.  They are also mathematically
easier to handle than general boolean networks.  Although structurally
quite simple, we show in Section \ref{sec:acyclic}  that they already
are, very often, more robust than simple static assignments.  In
Section \ref{sec:acycdef}, we start by formally defining
acyclic boolean networks and characterize them explicitly
algebraically.  Then, in Section \ref{sec:dectrees}, we discuss a
connection between acyclic networks robust with respect to an
objective function and decision trees for that function.  Using this
relationship, we give a procedure for constructing the
optimally robust acyclic boolean network with respect to an objective
function in time quasipolynomial in the
truth-table size of the function\footnote{In fact, as discussed in
  Section \ref{sec:algo}, efficient constructibility requires that the
optimally robust network be found in time polynomial in $n$, not the truth-table size of $2^n$.  Using
the connection with decision trees, we show efficient constructibility
for symmetric functions, which can be represented concisely.}. In
Section \ref{sec:neces}, the algebraic characterization from Section \ref{sec:acycdef} is used to
show that if a low-degree acyclic network is robust with respect to an
objective function, then the configurations expressed by the network
lie with high probability in an efficiently learnable subset of the
phenotype domain.  This is interesting because it suggests
heuristically that the need to achieve robustness necessarily
constrains nature to generating phenotypes that lie in an  ``easily
describable'' set.  This contrasts with the  
usual belief that biological structures are somehow very complex and
inexplicably varied.
In Appendix \ref{sec:ptf}, we show that the algebraic
characterization can be used to describe a large class of function families
robustly expressible by constant degree, acyclic boolean networks.
These functions are described as polynomial threshold functions, which
are well-studied objects in computer science (see
\cite{saksptf} for a survey).  

In Section \ref{sec:separ}, we show separation between various classes
of acyclic networks in terms of their ability to robustly express
objective functions.  First, we show explicit objective
functions which can be robustly expressed by (even low degree) acyclic
networks but are far from being robustly expressible by a static
assignment.  Then, we prove that acyclic boolean
networks of sub linear degree can robustly express only a tiny fraction of
functions of constant density.  On the other hand, a random function
of constant density is expected to be robustly expressed by an acyclic
network with no degree restriction.  Hence, since we know that in
biology, many regulatory networks are of low degree, this suggests
that either  the objective functions are chosen from the
``low-complexity'' set of functions we characterized earlier or that
the networks in biology are cyclic.

In Section \ref{sec:cyclic}, we investigate cyclic networks (networks that potentially possess
feedback loops).  We are
interested in networks that always arrive at fixpoints, even in the presence of
mutations, but  when started from specific initial configurations.
We show that cyclic networks can constrain the fixpoint configurations
much more than acyclic networks.  A little more precisely, one
implication of our results in Section \ref{sec:cyclic} is that cyclic
networks can force the fixpoint configurations to lie in a set of
$2^{O(\frac{n}{\log n})}$, in contrast to acyclic networks which can only robustly express functions that are
satisfied by at least $2^{\Omega(n)}$ configurations (see Section \ref{sec:neces}).
This gap separates cyclic from acyclic networks. We also show
that any objective function which can be satisfied by fixing the values
of a few variables is also robustly expressible by cyclic networks.
Conceptually, this means that if the genome needs some of its genes
to have some fixed expression levels, then it can
expend only a logarithmic number of regulatory genes in order to make
them fixed with high probability.  We also obtain stronger results
which show, for instance, that robust expressibility is guaranteed
whenever we have a set of variables of size $O(\frac{n}{\log n})$ that can
compensate for any mutations in the remaining variables.  These examples 
stand in strong contrast to the negative results for acyclic networks and suggest the
combinatorial power of cyclic networks.   \\
\section{Results and Proofs\label{sec:results}}

\subsection{Mathematical Preliminaries}

The mathematical notation used in this paper is fairly standard.  For
integer $n \geq 1$, we use $[n]$ to denote the set $\{1,\dots,n\}$.
Below, we explain the Landau notation, and then we describe two
results on random variables that are used frequently in the following
sections.

\subsubsection{Landau Notation\label{sec:landau}}

Big O notation (i.e. $O$, $o$, $\Omega$, $\omega$) or Landau notation describes the limiting behavior of a function when the argument tends towards a particular value or infinity. Throughout this paper it is used to describe the asymptotic behavior of quantities depending on the number of network nodes $n$ with $n$ increasing towards infinity. In the rest of the paper we omit  the specification $n \rightarrow \infty$. The formal definitions are as follows:

$$f \preceq g \Longleftrightarrow f(x) = O(g(n)) \Longleftrightarrow \limsup_{n\to \infty} \left|\frac{f(x)}{g(x)}\right| < \infty$$

$$f \prec g \Longleftrightarrow f(x) = o(g(n)) \Longleftrightarrow \limsup_{n\to \infty} \left|\frac{f(x)}{g(x)}\right| = 0$$

$$f \succeq g \Longleftrightarrow f(x) = \Omega(g(n)) \Longleftrightarrow \limsup_{n\to \infty} \left|\frac{f(x)}{g(x)}\right| > 0$$

$$f \succ g \Longleftrightarrow f(x) = \omega(g(n)) \Longleftrightarrow \limsup_{n\to \infty} \left|\frac{f(x)}{g(x)}\right| = \infty$$

Intuitively e.g. $f(n) = O(g(n))$ means that the quantity $f$ is (asymptotically) not bigger than ($\preceq$) than $g$ for large networks (large $n$). All these comparisons do not take constant factors into account. 

\subsubsection{Union bound}

The Union bound is the simplest way to bound the probability of an event $X$ which can not occur without one other event $X_1, X_2, ...$ taking place, i.e. $X \subseteq X_1 \cap X_2 \cap \ldots$. Very intuitively the probability that event $X$ occurs is  bounded from above by the sum of the probabilities that one of the other events occurs, precisely:

$$P(\bigcup_i X_i) \leq \sum_i P(X_i)$$

\subsubsection{Chernoff bounds}
 
The Chernoff bound is a concentration bound for the sum of independent
variables. It states that the probability that the sum of outcomes of
independent random experiments deviates from its expectation by $\eps$
decreases exponentially in $\eps$. There are multiple forms of this
bound. The ones used throughout this paper are:

\begin{theorem}
Let $X = \sum_{i\in [n]}X_i$  where $X_i$ for each $i \in [n]$ are
independently distributed in $[0,1]$.  Then, for $\eps \in (0,1)$:
\begin{itemize}
\item $\displaystyle P\left(X > (1+\eps)\Exp[X]\right) \leq \exp\left(-\frac{\eps^2}{3}\Exp[X]\right)$
\item $\displaystyle P\left(X < (1-\eps)\Exp[X]\right) \leq \exp\left(-\frac{\eps^2}{2}\Exp[X]\right)$
\end{itemize}
\end{theorem}

\subsection{Acyclic Networks\label{sec:acyclic}}

We start our investigation of the robustness property by restricting
attention to a subclass of boolean networks.  Naturally, because the
model is weaker, we can obtain stronger results here than we can in the
more general setting of unrestricted boolean networks to which we
return later.

\subsubsection{Definition and Characterizing Properties of Acyclic Networks\label{sec:acycdef}}
We are interested in networks that reach fixpoints starting from
an initial configuration and after arbitrary mutations.  To simplify
our task, let us restrict ourselves to networks that reach the same
fixpoint regardless of the initial configuration and such that
fixpoints reached by different mutations  are different.  More
formally, we say that a boolean network $N$ is {\em feed-forward} if:
\begin{itemize}
\item[(i)]
for any mutation $N'$ of $N$ and for any initial configurations
$\alpha_1, \alpha_2 \in \pcube^n$, $\Fix(N',\alpha_1) =   \Fix(N',
\alpha_2) \defeq \Fix(N')$
\item[(ii)]
for any two non-identical mutations $N'$ and $N''$ of $N$, $\Fix(N')
\neq \Fix(N'')$
\end{itemize}
Thus, in feed-forward networks, the initial configuration is
irrelevant in determining the fixpoint, and there is a bijection between
mutations of a network and their fixpoints.  

Our main object of study in this section is a special type of
feed-forward network.  We say a boolean network $N$ is {\em acyclic}
if any network $N'$ with $G_{N'} = G_N$ is feed-forward.  So, in particular, any mutation
of an acyclic network is also acyclic.  The name arises from the
following simple claim:
\begin{theorem}
A boolean network $N$ is acyclic iff $G_N$ is a directed acyclic graph.
\end{theorem}
\begin{proof}
If $G_N$ is a directed acyclic graph, then $N$ is feed-forward because one can update
the nodes of $N$ in sequence determined by a topological order on
$G_N$.  It is clear then that the stable configuration 
reached is a fixpoint, independent of initial configuration.  To see
that two non-identical mutations reach different fixpoints, consider
the first node in the topological order that is mutated in one network
but not in the other, and observe that their states in the fixpoint
configurations of the two networks will be different.\\
If $G_N$ contains a directed cycle, one can choose update functions
for the nodes that make the induced network have a cycle attractor.
This can be done simply by making sure that there is no assignment to
the nodes that satisfies all the update functions on the cycle
simultaneously.  
\end{proof}\\
Next, we show that by viewing the effect of mutations more
algebraically, we can precisely describe the structure of objective
functions which can be robustly expressed by acyclic networks.   In
doing so, we get a better understanding of how the degree of the
network limits the class of functions robustly expressible and how
fixpoints of the mutations of acyclic networks are distributed.  
First, some notation.  Define the $\eps$-biased product measure
$\mu_\eps$ on $\pcube^n$ by $\mu_\eps(x_1,\dots,x_n) = \eps^{n-k} 
(1-\eps)^k$ where $k = |\{i : x_k = 1\}|$.  We may view a function $f : \pcube^n \to \pcube$ as
the characteristic function of a subset of $\pcube^n$, that is, the subset $\{x \in \pcube^n: f(x) =
1\}$.  Then, $\mu_\eps(f)$ denotes the weight assigned by the measure $\mu_\eps$ to the set characterized by $f$.
Our main observation is the following.
\begin{lemma}\label{lem:main}
$f:\pcube^n \to \pcube$ is $\eps$-robustly expressible by an acyclic boolean network of degree $d$ if and only if
there exist $\pi, g, \varphi_1,\dots,\varphi_n$ such that:
\begin{equation}\label{eqn:basis}
f(x_1,\dots,x_n) = g(x_{\pi(1)}\cdot \varphi_1(), x_{\pi(2)}\cdot \varphi_2(x_{\pi(1)}), 
\dots, x_{\pi(n)}\cdot \varphi_n(x_{\pi(1)},\dots,x_{\pi(n-1)}))
\end{equation}
where:
\begin{enumerate}
\item[(i)] $\pi: [n] \to [n]$ is a permutation,
\item[(ii)] for $i \in [n]$, $\varphi_i : \pcube^{i-1} \to \pcube$ is a boolean function depending on
  at most $d$ inputs, and
\item[(iii)] $g: \pcube^n \to \pcube$ such that $\mu_\eps(g) \geq 1 - o(1)$.
\end{enumerate}
\end{lemma}
\begin{proof}
To prove one direction, suppose $f$ is $\eps$-robustly expressed by an
acyclic network $N$ of degree $d$.  Since $G_N$ can be topologically
ordered, there exists a permutation $\pi : [n] \to [n]$ such that
there is an edge between node $i$ and node $j$ in $G_N$ only if
$\pi^{-1}(x) \leq \pi^{-1}(y)$.  For every $i \in [n]$, let
$\varphi_i$ denote the update function associated with node $\pi(i)$
in the network.  Note that for any $i$, the function $\varphi_i$ can
only take as arguments at most $d$ elements of the set
$\{x_{\pi(j)}\}_{j \leq i}$.   Let $g(s_1,\dots,s_n) \defeq
f(x_1,\dots,x_n)$ where inductively, $x_{\pi(i)} = s_{\pi(i)} 
\varphi_i(x_{\pi(1)},\dots,x_{\pi(i-1)})$ for each $i \in [n]$.  One can explicitly verify now that
Equation (\ref{eqn:basis}) holds for this choice of $g$: If $s_{\pi(i)} = x_{\pi(i)} \cdot
\varphi_i(x_{\pi(1)}, x_{\pi(2)},\dots,x_{\pi(i-1)})$, then by our choice of $g$, $g(s_1,\dots,s_n) =
f(x_1,\dots,x_n)$.  Now, in an $\eps$-mutation of $N$, each $x_{\pi(i)} \cdot
\varphi_i(x_{\pi(1)}, x_{\pi(2)},\dots,x_{\pi(i-1)})$ is {\em
  independently} $1$ with probability $1-\eps$ and $-1$ with
probability $\eps$.  By definition of $\eps$-robust expressibility,
$g: \pcube^n \to \pcube$ is such that $\mu_\eps(g) \geq 1 - o(1)$.   

The proof in the other direction is similar.  Given the permutation $\pi$ and the functions
$\varphi_1,\dots, \varphi_n$, simply define a boolean network $N$ where $\pi$ gives the ordering of
the nodes and the $\varphi_i$'s specify the update functions of the nodes.  Then, the condition on $g$
ensures that $f$ is robustly expressed by the network.
\end{proof}

\subsubsection{Optimal Networks from Decision Trees \label{sec:dectrees}}

Furthermore, as we show next, we can construct the optimally robust
acyclic network for a given objective function in time quasipolynomial
in the truth-table size of the function.  For this, let us recall the
notion of a {\em decision tree} for a boolean function $f : \pcube^n
\to \pcube$.  It is a rooted binary tree $T_f$ where each edge is
labeled $-1$ or $1$, each non-leaf vertex is labeled with a
variable, and each leaf vertex is labeled with a $-1$ or $1$.  The
decision tree $T_f$ computes $f$ in the natural way: any assignment to
the variables determines a unique path from the root to a leaf, and
the label of the leaf at the end of the path is the value of the
function applied to the assignment.  

Now, with any decision tree $T_f$ for an objective function $f$, we
associate a boolean network $N_{T_f}$.  Given a mutation
parameter $\eps \in (0,\frac{1}{2})$, the decision tree is first preprocessed as
follows.  At each non-leaf node $v$ of the decision tree, we associate a real number $s(v)
\in [0,1]$ and a bit $b(v) \in \pcube$.  For a leaf node $v$, $s(v)$
is defined to be equal to $1$ if the leaf label is $1$ and $0$
otherwise.  For a non-leaf node $v$ with its two child nodes $v_1$ and
$v_2$, $b(v)$ is equal to the label of the edge leading to the child
$\arg \max_{w \in \{v_1,v_2\}} s(w)$, and $s(v)$ is equal to $(1-\eps)
\max(s(v_1), s(v_2)) + \eps \min(s(v_1), s(v_2))$.  That is, we define
$s(v)$ and $b(v)$ iteratively from the leaf nodes up to the root.
Now, the boolean network $N_{T_f} = (x,u_1,\dots,u_n)$ is defined by
setting each update function $u_i$ to output the value $b(v)$, where
$v$ is the first node in $T_f$ labeled $x_i$ obtained by following
edges of the decision tree down from the root.  (If $x_i$ is not
reached, then $u_i$ can be arbitrary.)  The degree of the network
constructed thus is at most $n$.  It is clear that if $T_f$ is a {\em
  layered decision tree}, i.e. nodes at the same distance from the
root have the same label, then the network $N_{T_f}$ is acyclic.   Figure \ref{fig:dectree}
illustrates the construction of the network $N_{T_f}$ from the decision
tree $T_f$ for the objective function $f(x_1,x_2,x_3) = (x_1 \wedge x_2) \vee (\overline{x_2} \wedge x_3)$.  

\begin{figure}
\begin{center}
\includegraphics[width=15cm]{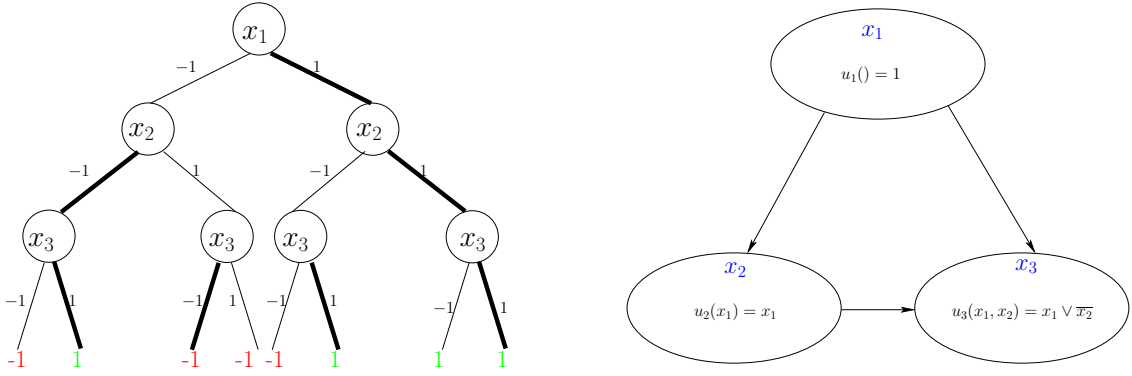}
\end{center}
\caption{For the objective function $f(x_1,x_2,x_3) = (x_1 \wedge x_2) \vee (\overline{x_2} \wedge x_3)$,
the decision tree is shown on the left.  Each node $v$ in the tree is associated with a value $b(v) \in \pcube$, as
described in the text, and the adjacent outgoing edge from $v$ labeled $b(v)$ is in bold.  The corresponding boolean
 network is shown on the right.}
\label{fig:dectree}
\vspace{.5cm}
\end{figure}

\begin{theorem}\label{thm:opt}
For a given objective function $f : \pcube^n \to \pcube$ and $\eps \in
(0,\frac{1}{2})$, the acyclic boolean network that is optimally
$\eps$-robust with respect to $f$ is a network $N_{T_f}$ for some
layered decision tree $T_f$ for $f$.
\end{theorem}
\begin{proof}
Fix a total order among the variables $x_1,\dots,x_n$.  We show that
if $T_f$ is a layered decision tree for $f$ reading the variables in
the selected order, then $N_{T_f}$ is the optimally $\eps$-robust
network with respect to $f$ among those acyclic networks $N$ for which
the DAG $G_n$ is consistent with the selected ordering.  To see this,
use induction on $n$.  If $n=1$, simply outputting the bit which
satisfies $f$ is optimally robust.  For $n>1$, consider the first node
in the total order.  In the optimally $\eps$-robust network, it must
be that after the first node has set its state, the remaining
network on $n-1$ nodes must also be optimally $\eps$-robust with respect to the
function on $n-1$ bits induced after setting the first node.  So, the
first node must set its state to the bit such that the network on
$n-1$ bits with the higher survival probability is chosen.  But this
is exactly how the network $N_{T_f}$ is constructed.
\end{proof}

Therefore, the time needed to construct the optimally robust acyclic network
with respect to a given objective function $f$ is at most $n! \cdot
O(2^n) \leq 2^{O(n \log n)}$, which is quasipolynomial in $2^n$, the
description length of an arbitrary boolean function.  For an objective function that is
guaranteed to have small description length, one could hope for a much
faster algorithm.  Our only result in this direction is the following
for {\em symmetric functions}; a function $f : \pcube^n \to \pcube$ is
said to be symmetric if for any permutation $\pi$ on $[n]$,
$f(x_1,\dots,x_n) = f(x_{\pi(1)},\dots,x_{\pi(n)})$.  
\begin{theorem}
The optimally robust acyclic network for a symmetric function can be
constructed in $O(n^2)$ time.
\end{theorem}
The key idea for the proof is to specialize the decision tree
algorithm for symmetric functions so as to reduce the number of queries.  

Notice that our construction of the network $N_{T_f}$ from the
decision tree $T_f$ did not depend on the fact that $T_f$ was
layered.  We term the networks arising from decision trees as {\em
  pseudo-acyclic} boolean networks.  One can easily check that any
pseudo-acyclic boolean network is feed-forward.  In fact, we
conjecture that any feed-forward network is also pseudo-acyclic, and so, in
some sense, pseudo-acyclic networks lie on the border between
acyclic networks and general boolean networks for which all mutations
reach fixed points (but perhaps starting only from certain initial configurations).  The optimally robust
pseudo-acyclic boolean network with respect to a 
given objective function can be found by enumerating over all decision
trees $T_f$ of a function and maximizing over the survival
probabilities of $N_{T_f}$.  The correctness argument is similar to that of
Theorem \ref{thm:opt}.  Most of our results in the following
subsections can be translated to the pseudo-acyclic setting.

\subsubsection{Robustly Expressible Polynomial Threshold Functions\label{sec:neces}}

In this section, we show connections between robust expressibility by
acyclic networks and {\em polynomial threshold functions}.  A function
$f : \pcube^n \to \pcube$ is said to be a polynomial threshold
function of degree $d$ iff $f$ can be written as $\sgn(p(x_1,\dots,x_n))$ where $p :
\pcube^n \to \mathbb{R}$ is a polynomial\footnote{Because the inputs are
$\pcube$, we can assume the polynomial to be multilinear without loss
of generality.} with real-valued coefficients of degree at most $d$
and where $\sgn: \mathbb{R} \to \pcube$ 
is the sign function which takes any negative input to $-1$ and any
non-negative input to $+1$.  Polynomial threshold functions are
well-studied objects in theoretical computer science, arising for
instance in learning theory and circuit complexity.  In particular,
low-degree polynomial threshold functions have been shown to be easy
to compute in several natural computational models.  

The characterization in Lemma \ref{lem:main} immediately leads to the
following implication for function families robustly expressible by acyclic
networks.  %Figure \ref{fig:ptf} illustrates the situation.

\begin{figure}
\begin{center}
\includegraphics[width=8cm]{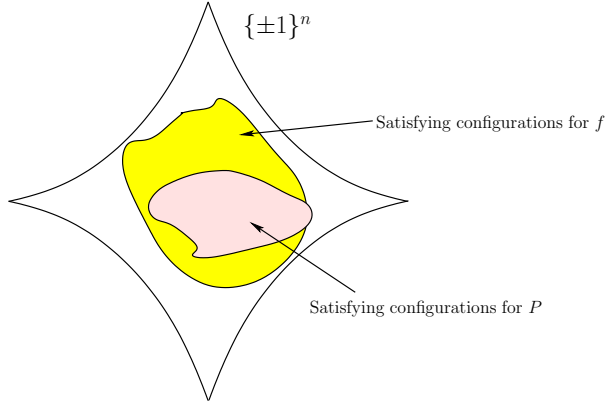}
\end{center}
\caption{Mutations of a constant-degree Boolean network robustly satisfying $f$ express
 configurations that, with high probability, lie in the solution set
 of $P$.  The solution set of $P$ is an efficiently learnable set,
 almost completely contained in the solution set of $f$.}
\label{fig:ptf}
\vspace{.5cm}
\end{figure}

\begin{theorem}\label{thm:impl}
If $f: \pcube^n \to \pcube$ is robustly expressible by an acyclic boolean network $N$ of constant
degree $d$, there is a function $P: \pcube^n \to \pcube$ such that:
\begin{enumerate}
\item[(i)]
$P$ is computable by a polynomial threshold function of degree $2d+2$,

\item[(ii)]
at least $1-o(1)$ fraction of the solution set of $P$ satisfies $f$, and

\item[(iii)]
the probability that an $\eps$-mutation of $N$ expresses a
configuration satisfying $P$ is at least $1-o(1)$.
\end{enumerate}
\end{theorem}
\begin{proof}
Suppose $f$ is $\eps$-robustly expressed by an acyclic network $N$ of
degree $d$.  We use Lemma \ref{lem:main} to write $f$ as
$g(x_{\pi(1)}\cdot \varphi_1(), x_{\pi(2)}\cdot \varphi_2(x_{\pi(1)}),
\dots, x_{\pi(n)}\cdot
\varphi_n(x_{\pi(1)},x_{\pi(2)},\dots,x_{\pi(n-1)}))$ where each
$\varphi_i$ is of constant arity $d$ and $g$ is such that $\mu_\eps(g)
\geq 1-o(1)$.  For $i \in [n]$, set
$s_i = x_{\pi(i)} \cdot \varphi_i(x_{\pi(1)},\dots,x_{\pi(i-1)})$.
As we showed in the proof of Lemma \ref{lem:main}, there is a
bijective correspondence between $(x_1,\dots,x_n)$ and
$(s_1,\dots,s_n)$.  

Let $R = \left\{(x_1,\dots,x_n) \in \pcube^n: \sum_{i=1}^n s_i \in [(1-3\eps)n,
  (1-\eps)n]\right\}$.  We will construct $P$ so that its solution set
is $R$. In a configuration expressed by $\eps$-mutation of $N$, each
$s_i$ is independently $+1$ with probability $1-\eps$ and $-1$ with probability $\eps$.
Therefore, by standard Chernoff bounds, a configuration $(x_1,\dots,x_n)$ expressed by an 
$\eps$-mutation of $N$ satisfies the property that $|\sum_i s_i - (1-2\eps)n|
< \eps n$ with probability at least $1-o(1)$, proving part (iii).  Moreover, because $\mu_\eps(g)
\geq 1-o(1)$, it follows that $\Pr_s[g(s_1,\dots,s_n) = 1] \geq
1-o(1)$ where $s = (s_1,\dots,s_n)$ is drawn uniformly from the set 
$\left \{(s_1,\dots,s_n) \in \pcube^n : (1-3\eps)n \leq \right.$
$\left. \sum_i s_i \leq (1-\eps)n]\right\}$, proving part (ii).  Finally, for (i), observe
that $R$ is the solution set for $$P(x_1,\dots,x_n) = {\rm sgn}\left[\left(\sum_i
  s_i - (1-3\eps)n\right)\left((1-\eps)n - \sum_i s_i\right)\right]$$
Expressing each $\varphi_i$ as a multilinear polynomial of degree $d$,
each $s_i$ becomes a multilinear polynomial of degree $d+1$.  Therefore,
$P$ is the sign of a $(2d+2)$-degree polynomial, proving our theorem. 
\end{proof}

As stated below in Corollary \ref{cor:deg}, Theorem \ref{thm:impl}
implies that a degree bound on an acyclic boolean network means that
it can only robustly express those functions which contain a
``simple'' subset.  In other words, no matter how complicated the
objective function $f$ is, a low-degree acyclic network robustly
expressing $f$ maintains assignments that solve a ``simpler''
subfunction of $f$.  The biological implication is that if there is a
population of organisms trying to satisfy the same environmental
constraint, then the phenotypes displayed in the population can be
efficiently described.  This conclusion goes against the commonly held
belief that phenotypes occurring in biology are varied in a very
complicated fashion.  

To state the corollary precisely, we use the notion of {\em PAC
  (probably approximately correct)   learnability}, the most commonly
used theoretical framework in machine learning.  Essentially, a
function is said to be PAC-learnable if there is an efficient
algorithm that can use random example evaluations of the function with
respect to some probability distribution to learn the function with
high probability on a large fraction of the domain; see
\cite{kearns1994icl}, for instance, for more details.   

\begin{corollary}\label{cor:deg}
For a function $f : \pcube^n \to \pcube$ and $\eps \in (0,\frac{1}{2})$, if the
solution set of $f$ does not contain a set $S$ of size at
least\footnote{$H(\cdot)$ is the binary entropy function: for $p \in
  (0,1)$, $H(p) = -p \log_2 p - (1-p) \log_2 (1-p)$.}
$2^{H(\eps)n}$, that is PAC-learnable in polynomial time
with respect to the uniform distribution, then there is no acyclic
constant-degree boolean network that $\eps$-robustly expresses $f$.
\end{corollary}
\begin{proof}
Suppose otherwise.  Let $R$ be the solution set of the function $P$
guaranteed by Theorem \ref{thm:impl}, and let $S = R \cap \{x=(x_1,\dots,x_n) :
f(x) = 1\}$.  By the proof of Theorem \ref{thm:impl}, the size of $S$
is at least $(1-o(1))\sum_{w = (1-3\eps)n}^{(1-\eps)n} \#\{x \in \pcube^n :
\sum_i x_i = w\} \geq 2^{(H(1.5\eps) - o(1))n}$ by Chernoff bounds.  Moreover, $S$ can be
$(1-o(1))$-approximated with respect to the uniform distribution by a
constant-degree polynomial threshold function, which can be
PAC-learned from random examples in polynomial time
\cite{KlivansServedio04,KlivansOdonnellServedio04}.  
\end{proof}
\vspace{.3cm}

\subsubsection{Dependence of Robustness on the Network Degree \label{sec:separ}}

In the next two sections, we show results in different directions that illustrate
how larger degree networks can robustly express more objective functions.\\[1em]
\paragraph{Advantage of Networks over Static Assignments\label{sec:badstat}}

The class of objective functions shown in Appendix \ref{sec:ptf} to be robustly expressible
provide examples of cases in which acyclic networks are strictly stronger than static assignments 
(degree-$0$ networks) in the sense that they are able
to robustly express functions which are out of reach for static
assignments.

\begin{corollary}\label{cor:diffstat}
For any constant $\eps \in (0,1/2)$, there is a family of functions $f_n : \pcube^n \to \pcube$ such that
it is $(\eps,1-2^{-\Omega(\sqrt{n})})$-robustly expressible by a
degree-$2$ boolean network but for which there does not exist a static
assignment with survival probability $2^{-o(\sqrt{n})}$.   
\end{corollary}
\begin{proof}
For each $n$, consider $g_n(x_1,\dots,x_n) = \sgn(x_1x_2 + x_1x_3 + \cdots + x_1x_n)$.  $g_n$ 
satisfies the conditions of Theorem \ref{thm:tree-suff} and hence is robustly expressible by a
degree-$2$ boolean network.  On the other hand, for any static
assignment, $x_1$ could be assigned to the complement of $\sgn(x_2 +
\cdots + x_n)$ with constant probability $\eps$, so that an $\eps$-mutation of the assignment would
not satisfy $g_n$ with probability $1-\eps$.

Now let 
\begin{align*}
f_n(x_1,\dots,x_n) = \sgn(- (1-4\eps)\sqrt{n} &+ g_{\sqrt{n}}(x_1,\dots,x_{\sqrt{n}}) \\
												  										&+ g_{\sqrt{n}}(x_{\sqrt{n}+1},\dots,x_{2\sqrt{n}}) \\
										 													& \cdots \\
																							&+ g_{\sqrt{n}}(x_{n-\sqrt{n}+1}, \dots,x_n)\ )
\end{align*}
A Boolean network robustly expressing $f_n$ is simply the disjoint
union of the $\sqrt{n}$ Boolean networks expressing each of the $g_{\sqrt{n}}$'s.  The probability
that an $\eps$-mutation of the network expressing $f_n$ is $1-2^{-\Omega(\sqrt{n})}$, by the
Chernoff bound.  For a static assignment, we argued above that a static assignment can express
$g_{\sqrt{n}}$ with probability at most $1-\eps$; hence, the expected value of $g_{\sqrt{n}}$ with
respect to $\eps$-mutations is $1-2\eps$.  Again, by the Chernoff bound, the survival probability of
a static assignment for $f_n$ is then at most $2^{-\Omega(\sqrt{n})}$. 
\end{proof}

In the case where we don't care about the degree we can even show
exponentially small survival probability for a static assignment on
some functions which are robustly expressible by acyclic networks. (See Appendix \ref{sec:strongsep}.)

\vspace{.3cm}
\paragraph{Most Functions need Full Degree Networks \label{sec:random}}

We show in this section that random functions can not be robustly expressed by (pseudo) acyclic networks
of bounded degree. We give proofs and interesting evidence that unbounded degree acyclic networks lie
very close to the boundary of expression power needed to express the vast majority of functions.
For $\rho \in (0,1)$ and $n$ a positive integer, let ${\cal F}_{n,\rho}$ denote the distribution on functions mapping
$\{\pm 1\}^n$ to $\{\pm 1\}$, induced by letting each entry of the truth table of the function be $1$ with probability $\rho$ and $-1$ with probability $1-\rho$.  

\begin{theorem}
For constants $\eps \in (0,\frac{1}{2})$ and $\rho \in (0,1)$, with
probability at least $1-o(1)$, there is no network $N$ of degree
$o(n)$ that $\eps$-robustly expresses a function chosen uniformly at
random from ${\cal F}_{n,\rho}$.
\end{theorem}
\begin{proof}
Fix a network $N$ with maximum degree $d$.  If a boolean function $f$
is $\eps$-robustly expressed by $N$, then by Corollary \ref{cor:deg}
$f$ must be satisfied on a set of size at least\footnote{$H(\cdot)$ is the binary entropy function: for $p \in
  (0,1)$, $H(p) = -p \log_2 p - (1-p) \log_2 (1-p)$.} $2^{H(\eps)n}$.
Therefore, the probability that a function uniformly chosen at random
from ${\cal F}_{n,\rho}$ is satisfied on this set is at most
$\rho^{2^{H(\eps) n}}$.  The total number of pseudoacyclic boolean networks of
degree at most $d$ is at most $2^{2^d n} n^{O(n)}$.  Therefore,
applying the union bound:
\begin{equation*}
\Pr_{f \in {\cal F}_{n,\rho}}[\exists N \text{ of degree } d~
\eps\text{-robustly expressing} f] \leq \rho^{2^{H(\eps) n}} 2^{2^d n}
n^{O(n)} \leq o(1)
\end{equation*}
if $d = o(n)$.
\end{proof}

While the last theorem showed that acyclic networks with slightly less than full degree are not able to robustly express most functions the next theorem shows that without this small restriction acyclic networks can be found on the boundary of being capable to robustly express random functions with fixed density $\rho$. The next theorem shows a tradeoff between the density parameter $\rho$ which determines how hight the percentage of viable configurations is in expectation, the mutation parameter $\eps$ and the survival probability $\delta$.

\begin{theorem}\label{thm:randlb}
For any $\rho \in (0,1)$ and $\eps \in (0,\frac{1}{2})$, 
if a function $f$ is chosen uniformly at random from ${\cal
  F}_{n,\rho}$, then there is a boolean network $N$ that, in 
expectation, $(\eps,1-(1-\rho)^{\log \frac{1}{\eps}})$-robustly
expresses $f$.  In other words, for any survival probability $\delta
\in (0,1)$, there is a mutation parameter $\eps \in (0,1)$ such that a boolean network $N$ $(\eps,
\delta)$-robustly expresses $f$ in expectation.
\end{theorem}

The idea of the proof for Theorem \ref{thm:randlb} is to use our
procedure for constructing optimally robust acyclic networks from decision trees as described in Section
\ref{sec:dectrees} and then to lower-bound the success probability of the
resulting network.  We do not have a proof that the lower-bound is
tight, and so, it might even be possible that a random function from
$\calf_{n,\rho}$ is expected to be robustly expressible by an acyclic
network. 

In spite of the power of acyclic networks as demonstrated above, there
are still some important classes of objective functions that are not
robustly expressible by them.  The next theorem  gives some examples:
\begin{theorem} \label{thm:nonacyclic}
~\\[-2em]
\begin{itemize}
\item
If a function $f : \pcube^n \to \pcube$ is a $k$-junta for constant
$k$, i.e. $f$ depends only on at most $k$ variables, then the
optimally robust acyclic network with respect to $f$ has constant
success probability.  For example, dictator functions
($f(x_1,\dots,x_n) = x_i$ for some $i \in [n]$) are not robustly
expressible by acyclic networks.

\item
Suppose $f: \pcube^n \to \pcube$ is a symmetric function; then there
exists a function $g : [0,n] \to \pcube$ such that $f(x_1,\dots,x_n) =
g\left(\sum_i \frac{1+x_i}{2}\right)$ for all $x \in \pcube^n$.  If
$g$ has the property that for any constant sized interval $I \subset
[0,n]$, there exists $s \in I$ such that $g(s) = 0$, then the
optimally robust acyclic network with respect to $f$ has constant
success probability.  For example, parities
($f(x_1,\dots,x_n) = \prod_{i=1}^n x_i$) are not robustly expressible
by acyclic networks.
\end{itemize}
\end{theorem}

\subsection{Cyclic Networks \label{sec:cyclic}}

In this section we show that using the full power of cyclic networks gives significantly more robust functions. As a main result we give a construction of networks always converging to fixpoints which have $o(\frac{n}{\log n})$ variables nearly fixed in dependence on all other variables. This allows to robustly express dictator functions, many symmetric functions, $o(\frac{n}{\log n})$-juntas all functions shown (i.e. in theorem \ref{thm:nonacyclic}) to be way beyond the reach of acyclic networks. This shows that the dynamic behavior of acyclic networks can be used to stabilize the potential disruptions of highly critical parts of a regulatory network by random mutations while still allowing evolution and changes in other parts.

\begin{theorem}\label{thm:fix}
There is a cyclic network $N=(x,u)$ on $2T$ variables and a start configuration $y$ with the property that with probability $1 - c_\eps^{-T}$ an $\eps$-mutation $N'$ starting from the configuration $y$ converges to the fixpoint $Fix(N',y)$ in at most 3 steps and the value of a specified variable $x_1$ in the fixpoint is $1$. The value for $c_\eps$ depends on the mutation rate $\eps \in (0,\frac{1}{2})$ and is at least $e^{\frac{(1-2\eps)^2}{1-\eps}} > 1$. 
\end{theorem}
\begin{proof}
We explicitly describe the network. The network consists of the variable $x_1$ and a set $I$ of $|I|=2T-1$ indicator variables which are supposed to detect when the variable $x_1$ was affected by a mutation. Their update function gets $-1$ when $x_1=-1$ or when the majority of the indicator variables is $-1$. The value of the update function $u_1$ for the variable $x_1$ is the value of the majority of the indicator variables. The start configuration $y$ for the network is the all one vector.\\
For the analysis of this network we want to prove that for any network $N'=(x,v)$ which differs from $N$ by less than $T$ mutations a fixpoint $x'(\infty) = \Fix(N',y)$ with $x'(\infty)_1=1$ is reached. For this we see that the network starts at $t=0$ with the all zero configuration $x'(1)=y=1$ and that at the next time $t=2$ exactly all variables whose update function got mutated turn into $-1$. In the case where $u_1$ was not mutated $x_1$ is still $1$ and this is the fixpoint reached by the network. To see this observe that less than $T$ vote switches are not enough to switch the majority vote of the indicator variables. This leaves us with with the case that $u_1' = -u_1$ in which the value of $x_1$ at time $t=2$ is 1. For $t=3$ this results in all unmutated indicator variables getting 1. Since the unmutated indicator variables form a majority they stay $-1$ for all $t\geq 3$. This makes the value of the variable $x_1$ at any time $t>4$ to be $x'_1(4)=u'_1(x'_3)=-u_1(x'_3)=-(-1)=1$. Thus at $t=4$ for the desired fixpoint is reached.\\
Since all mutations occur independently with probability $\eps$ the probability that an $\eps$-mutation $N'$ differs from the network $N$ in less than $T$ mutations is exactly $\sum_{k=0}^{T-1} {2T \choose k} (1-\eps)^k \eps^{2T}$. By standard Chernoff-bound this is at least $1 - e^{- \frac{2((1-\eps)2T - T)^2}{(1-\eps) 2T}} = 1 - c_\eps$ with $c_\eps = e^{\frac{(1-2\eps)^2}{1-\eps}}$.
\end{proof}

Instead of just fixing one variable we can apply the above theorem multiple times and
show that any objective function which can be satisfied by fixing the values of a few variables can be robustly
expressed by a cyclic network.

\begin{theorem}\label{thm:cyclic}
Any boolean function which can be satisfied by fixing the values of $o(\frac{n}{\log n})$ variables is robustly expressible by a cyclic network.
\end{theorem}
\begin{proof}
By assumption there are $s = o(\frac{n}{\log n})$ variables and a partial assignment $y$ assigning each of
these variables a value such that every configuration which agrees
with $y$ on these variables is viable. For each of the $s$ variables
we use a network with $2T = \frac{n}{s} = \omega(\log n)$ variables as
described in Theorem \ref{thm:fix} to fix its value to the one given by $y$. This
results in a network $N$ of size $n$. With probability $(1 -
c_\eps^{-T})^s > 1 - n c_\eps^{-\omega(\log n)} = 1 - o(1)$ an
$\eps$-mutation $N'$ will have less than $T$ mutations in each
subnetwork and converge to a viable fixpoint.  
\end{proof}

Our next theorem substantially extends the range of robustly
expressible objective functions. It shows that it is possible to even robustly express any functions for which any assignment can be made satisfying by changing only a few variables. This is much stronger than the last theorem since it allows us to choose our values according to all other variables instead of fixing the values beforehand. 

\begin{theorem}\label{thm:maincyclic}
Let $f$ be a boolean function $f$ on the variable set $X = \{x_1,\ldots,x_n\}$ and $S = \{x_{i_1},\ldots,x_{i_{|S|}}\} \subset X$ be a subset of variables of size $s = |S| = o(\frac{n}{\log n})$. If for any partial assignment $x_{X\setminus S} = y$ of values to variables in $X \setminus S$ there is an assignment $x_{S,y}$ to all variables in $S$ which completes the partial assignment to a viable configuration, than $f$ is robustly expressible.
\end{theorem}
\begin{proof}
We directly construct the network expressing $f$. We always choose the
all-one-configuration as the start configuration. We divide the
variables in $X \setminus S$ in $s+1$ blocks $S_0,\ldots,S_{s}$ of
size $t = \frac{n - |S|}{|S| + 1} = \omega(\log n)$. For each
$j=1,\ldots,s$ the block $S_j$ is assigned to the variable $x_{i_j}$,
while the variables in $S_0$ are sync variables. Figure
\ref{fig:cyclic2} illustrates the partitioning of the nodes.
\begin{figure}
\begin{center}
\includegraphics[height=6cm]{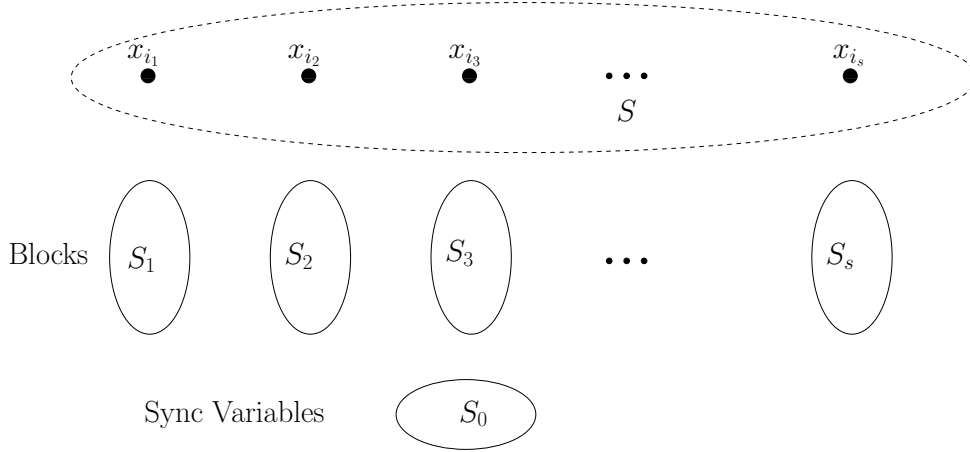}
\caption{Shows the partitioning of the nodes in the network
  in the proof for Theorem \ref{thm:maincyclic}.}
\label{fig:cyclic2}
\end{center}
\vspace{.5cm}
\end{figure}

The idea is that it is very unlikely that more than $\frac{t}{3}$ mutations occur in any of the blocks of size $t$. Having this in mind we want to proof that any network $N' \neq N$ which differs from $N$ in less than $\frac{t}{3}$ mutations in every block converges against a viable configuration.\\
The update function of every sync variable $x \in S_0$ is $1$ iff all variables in $X$ are $1$. For every $j=1,\ldots,s$ the update function for the variable $x_{i_j}$ is $1$ if every variable in $X$ is $1$ and else outputs the value of its variable in the assignment $x_{S,x_{X\setminus S}}$ times the value of the majority of variables in $S_j$ (i.e. inverting it if the majority is $-1$). Lastly for every $j=1,\ldots,s$ the update function of variables in a block $S_j$ are $-1$ if (the variable $x_{i_j}$ is $-1$ and the majority of $S_0$ is $1$) or if the majority of $S_j$ is $-1$.\\
This has the following effect: The all-one start configuration is stable in $N$ giving a $1$ in every unmutated variable for $t=2$ and a $-1$ in every mutated variable. Since $N \neq N'$ there is at least one variable $-1$ for $t=2$ giving a value of $-1$ for all unmutated variables in $S_0$ for $t=3$. Since the majority of variables in $S_0$ is not mutated the majority of $S_0$ is fixed to $-1$ for $t \geq 3$. The majority of variables in any other block $S_j$ is $1$ for $t=1$ and thus also for $t=2$ since $x_{i_j}(1) = 1$. At $t=3$ the all unmutated variables in $S_j$ update to the value of $x_{i_j}(2)$ which we have already argued is $-1$ iff this variable is mutated. Since the majority of $S_0$ is $-1$ for $t \geq 3$ the update function for the unmutated variables in $S_j$ reduced to being $-1$ iff the majority of them is $-1$. This inductively preserves the value of $x_{i_j}(2)$ for the majority of $S_j$ for all $t \geq 3$. Taking all this together shows that all variables in $X \setminus S$ are fixed for all $t \geq 3$. This gives fixed values for variables in $S$ for $t \geq 4$. Thus we have a fixpoint and we can verify indeed that for all unmutated variables $x_{i_j} \in S$ the majority of the corresponding block $S_j$ is $1$ and the value of $x_{i_j}$ in the fixpoint is the value in $x_{S,x_{X\setminus S}}$. If the variable $x_{i_j} \in S$ is mutated in $N'$ the majority of the corresponding block $S_j$ is $-1$ and the value of $x_{i_j}$ in the fixpoint is $-1 \cdot -1 \cdot x_{S,x_{X\setminus S}} = x_{S,x_{X\setminus S}}$ as well. By assumption the resulting fixpoint is a viable configuration.\\
This shows that if at most $\frac{t}{3}$ mutations occur in every of the $s+1$ blocks and at least one mutation occurs in total, the resulting network $N'$ reaches a viable configuration as a fixpoint. The probability that an $\eps$-mutation of $N$ does not fulfill this property is at most $(1 - c_\eps^{-\frac{t}{3}})^s > 1 - n c_\eps^{-\omega(\log n)} = 1 - o(1)$.
\end{proof}

The last two theorems prove the power of cyclic networks. They can robustly express many functions
which are out of reach for cyclic networks. To exemplify this we note that Theorem \ref{thm:fix} and 
respectively Theorem \ref{thm:fix} respectively allow us to fix values of one or $k \prec \log n$ variables. This allows us immediately to robustly express the dictator and $k$-junta functions which where shown to be non robustly expressable by acyclic networks at the end of Section \ref{sec:random}. Theorem \ref{thm:maincyclic} goes even further since it allows to fix variable values in hindsight. To demonstrate how big of an advantage this is, taking even $s=1$ suffices. Here Theorem \ref{thm:maincyclic} allows us to pick a variable, look at the development of all other values first and then choose how to set our variable. This allows us e.g. to robustly express parity and other more complicated functions robustly, even so their outcome depends on all variables. In fact while sublinear degree acyclic networks can not express most (randomly chosen) functions from $\calf_{n,\rho}$ even if $\rho$ is a constant (see Theorem \ref{thm:randlb}), Theorem \ref{thm:maincyclic} implies that with very high probability, a uniformly selected function from $\calf_{n,\rho}$ is robustly
expressibly by cyclic networks even with an exponentially low $\rho \geq 2^{-o(\frac{n}{\log n})}$. This means that in nearly all environments -- even with extremely sparse viable configurations -- robust cyclic networks exist and work reliably.

\section{Conclusion and Future Work}

In the absence of networks, robustness  can only be achieved if it is possible to set
the gene expression levels such that nearly all direct changes on these levels are still
viable. In different words, the static configuration of gene expression levels must be chosen
such that a large fraction of the of the configuration's Hamming neighborhood is viable.  What we have shown in this
paper is that using boolean networks, instead of just static configurations, allows us to shape this neighborhood induced by random mutations.  A
single mutation to the genotype can cause the phenotype configuration to change by a lot.  The choice of the update
functions in the boolean network defines now a differently shaped neighborhood of a phenotype configuration.  Studying acyclic networks, we
showed that the neighborhood induced by the update functions of an acyclic network is a bijective transformation of the usual
neighborhood in terms of Hamming distance.  This means that all the randomness or variability introduced by nature
is preserved but guided in the direction of still viable options. This is done by reshaping the neighborhood in such a way that most configurations in it are viable.  The greater the degree of the network, the more ``complex'' the bijection is allowed to be and thus the greater
possibility of robustly expressing objective functions.

In contrast to this cyclic networks can make the volume of the phenotype neighborhood smaller. That is, cyclic
networks can compress different mutations in genotypic space into the same change in phenotypic space.  This property of cyclic networks
makes them useful for fixing the value of nodes in the network - providing stability.
Intuitively, the reason cyclic networks can compress volumes in phenotypic space is that feedback loops provide nodes in the network
the power to detect whether they have been mutated. Using this feedback behavior of cyclic dynamics allows dramatically more powerful transformations and concentrating probability mass in a smaller volume of phenotypes helps to
protect highly critical parts of an organism.

Even though our work should be mostly understood as a first attempt to formalize and
better understand the nature of robustness, it nevertheless suggests nontrivial predictions
for biological systems. Our results in Section \ref{sec:cyclic} indicate that cyclic networks should be more
prevalent in biological systems where robustness is more important than variability.  For example, we
expect parts of an organism which are responsible for highly critical functions and survival to be regulated by
self-reinforcing feedback loops. The prevalence of cyclic networks should also occur in
general for organisms in rough environments allowing only specialized and well adjusted
organisms to survive. On the contrary, in parts of a biological system where evolvability
is highly desired, our results in Section \ref{sec:acyclic} indicate that there will be more
use of acyclic and pseudo-acyclic networks. This is for example the case in systems which
are friendly in the sense that they allow many viable phenotypes but are at the same time
rapidly changing. In such a situation the high evolvability of acyclic and pseudo-acyclic
networks maximizes the chance to adjust to changes by guiding the full randomness via the
control over their geometry in the phenotypic space, while their bad robustness behavior
does not harm the development.\\
Even though experimental research is still far away from determining the structure of large
regulatory networks, it is conceivable that the understanding that has been developed for a number of concrete small biological
systems will in the near future be available for a much larger class of regulatory networks. For the case that systematic experimental
evaluation will be able to extract explicit objective functions one could check whether they are robustly expressible by our definition. Even more
importantly, one could test the proposed model and its predictions by comparing the structure and robustness parameters of the networks our constructions gives (e.g. via the decision tree algorithm) with the actual networks occurring in nature.\\
Notwithstanding experimental validation, our work in this paper is
significant because it gives a rigorous language to examine the
robustness of the regulatory network.  Even if different notions of
mutation or different genotype-phenotype maps are used in other models
of the regulatory network, the basic questions asked in this paper are
still relevant and we suspect that for most reasonable setups, the
answers will be similar. Furthermore, the very idea that a robust
network could be designed  and  proved to be
robust (with respect to a given mutation model and genotype-phenotype
map) was, as far as we know,  not made explicit in previous work and
is an important conceptual contribution of our work.  This new line of
research leaves open several unresolved questions which we discuss next.

\subsubsection*{Future Work}

Beyond the results in this paper, we are interested in a better understanding of cyclic networks. How powerful are they in concentrating
the results of mutations in a smaller volume of the phenotype space? Mathematically
this question translates directly to the quest for upper bounds on the probability mass
$P(Fix(N_\eps,\alpha) \in S)$ that can be concentrated in an asymptotically small set $S$.
We also would like to get more examples of classes of objective functions for which robust networks 
are efficiently constructible, like the acyclic networks for symmetric functions. 

In future work it would also be interesting to analyze alternate mutation models, including ones where the boolean
network can be changed more drastically than here.  Is it still possible to obtain any robustness in such scenarios? Can similar results, gaps and tradeoffs between, robustness, network degree and structure be made?

Lastly we would like to advocate similar rigorous studies of other central ideas in biology.  For example, it would be
interesting to understand more precisely the ideas proposed in \cite{wag05b} that robustness promotes evolvability.  Can we develop a
formal model where we show rigorously that robustness to one objective function helps a system find solutions to related objective functions?

\section*{Acknowledgements}
A.B. would like to thank Gerald J. Sussman for introducing him to the subject and for sparking
an interest in formalizing biological concepts.  He would also like to thank Madhu Sudan and Victor 
Chen for valuable early conversations on a related project.  A substantial fraction of the work was 
completed while A.B. was at an internship at Sandia Laboratories, Livermore;
he would like to thank Rob Armstrong and Jackson Mayo for constructive and illuminating discussions and constant encouragement.
Finally, Manolis Kellis provided useful feedback during the latter stages of
the work.

\bibliographystyle{alpha}
\bibliography{robustness}

\appendix

\section{Classes of Robustly Expressible Functions\label{sec:ptf}}

Although Lemma \ref{lem:main} provides a tight condition for robust expressibility, it is not very natural, in the sense that the
condition cannot be verified easily unless the function is presented in a particular form.   In this section, we derive some more natural
conditions that guarantee robust expressibility.

We start with a useful definition.
\begin{definition}[Sequential Cover]
A bipartite graph $G = (V_1, V_2, E)$ with $|V_1| = m$ and $|V_2| = n $ is {\em sequentially
  coverable} if there exists a sequence of vertices $v_1,\dots,v_k \in V_2$ for some $k \leq n$ such
that the following two conditions hold: 
\begin{enumerate}
\item[(i)]
Every vertex $v \in V_1$ is a neighbor of some $v_i$
\item[(ii)]
Let $G_0 = G$.  For $i \in [k]$, inductively define $G_i$ as the induced graph on $G_{i-1}
\backslash (\{v_i\}\cup {\cal N}(v_i))$.  Each $v_i$ is a vertex of degree exactly $1$ in $G_{i-1}$.
\end{enumerate}
The sequence $v_1,\dots,v_k$ is called a {\em sequential cover of size $k$} for $G$.
\end{definition}

A bipartite graph is thus sequentially coverable if the vertices of
$V_2$ can be ordered in such a way that at most one vertex of $V_1$ is
covered at a time.  Note that $m \leq n$ necessarily if the graph is
sequentially coverable. Also, given a polynomial $p : \re^n \to \re$,
let $G_p$, the {\em term-variable   graph of $p$}, be the bipartite
graph with vertices for every variable and for every term, and with an
edge between a variable vertex and a term vertex iff the variable
occurs in the term.

\begin{theorem}\label{thm:main-neces}
If a function $f : \pcube^n \to \pcube$ has a sign-representation $\sgn(p(x_1,\dots,x_n))$, such
that $p(x)$ is a degree-$d$ polynomial with constant coefficients and such that its term-variable bipartite graph, $G_p$, has a
sequential cover of size $\Omega(n)$, then $f$ is robustly expressible by a boolean network of
degree $d-1$.
\end{theorem}
\begin{proof}
Let $f$ be a function of the variables $X = \{x_1,\dots,x_n\}$.  We construct a boolean net $N$ that
$\eps$-robustly expresses $f$ for some constant positive $\eps$.  Suppose that the term-variable
graph $G_p$ is sequentially covered by the sequence of variables $x_{i_1},\dots,x_{i_k} \in X$,
where each $i_j$ is a distinct element of $[n]$.  Let $x_{i_{k+1}},\dots,x_{i_n}$ denote the rest of
the variables (in some arbitrary order).  For $j \in [k]$, let $T_j$ denote the unique term covered
by the variable $x_{i_j}$.  Observe that for $j \in [k]$, $T_j$ can only contain the variables
$\{x_{i_\ell}\}_{\ell \geq j}$ and always contains $x_{i_j}$.  In the boolean network $N$, let the
update function for the node associated with $x_{i_j}$ be $u_{i_j} = \sgn(T_j/x_{i_j})$ for $j \in
[k]$ and let $u_{i_j}$ be an arbitrary element of $\pcube$ for $j \in \{k+1,\dots,n\}$.  It is clear
that $N$ is an acyclic boolean network.

We now show that $f$ is $\eps$-robustly expressed by $N$ for some
$\eps \in (0,\frac{1}{2})$.  Observe that with probability at least 
$1-2^{-\Omega(n)}$, at most $2\eps n$ mutations occur.  There are a total of $\Omega(n)$ terms in
$p$.  The terms which correspond to mutated nodes are strictly negative, while those which are not are
strictly positive, because of our choice of update functions.  Since all the coefficients of $p$ are
constant, for a small enough constant $\eps$, at most $2\eps n$ mutations will not be enough to make
$p$ evaluate to a negative real.  Hence, $N$ expresses $f$ with probability at least
$1-2^{-\Omega(n)}$.  
\end{proof}\\
We will say that a sign-representation is {\em acyclic} if the term-variable graph contains no
cycle.  This allows us to present a more natural class of functions that are robustly expressible.\\
\begin{theorem}\label{thm:tree-suff}
If a boolean function $f : \pcube^n \to \pcube$ has an acyclic constant-degree sign-representation
with constant coefficients and no degree-$1$ terms, then $f$ is
robustly expressible by an acyclic boolean network of constant degree. 
\end{theorem}
\begin{proof}
We show that $f$ has a sign-representation whose term-variable graph has a sequential cover of size
$\Omega(n)$, thus proving our claim using Theorem \ref{thm:main-neces}.  Let $G$ be the 
term-variable bipartite graph for the given sign-representation for $f$.  Since $G$ is a forest by
assumption, there must exist some (at least $2$) degree-$1$ vertices.  Furthermore, because there
are no degree-$1$ monomials in the 
sign-representation, all the degree-$1$ vertices represent variables, not terms.  We construct $S$, a
sequential cover of $G$, as follows.  Initially $S$ is empty.  Select some degree-$1$ vertex $v$ in
$G$ and append it to $S$. Next, remove $v$ from $G$ and also all the vertices adjacent to $v$.  Note
that these adjacent vertices must represent terms.  The modified graph is still a forest and must
have some degree-$1$ vertices.  Again, the degree-$1$ vertices must represent variables, not terms.
Hence, we can repeat the process, appending a degree-$1$ vertex to $S$, remove it and its adjacent
vertices from $G$, and so on.  We stop when no vertices remain that represent terms.  \\
It is clear that $S$ is a sequential cover.  We only need to show that $S$ is of size $\Omega(n)$.
This is so because each time a vertex is added to the sequential cover, we remove the unique term
the associated variable is contained in, and this removal can make only a constant number of other
vertices isolated (because each term is of constant degree).  Hence, in order for all the variable
vertices to either be in $S$ or be isolated after the removal process, at least $\Omega(n)$ vertices
need to be in $S$. 
\end{proof}

\section{Follow-up to Corollary \ref{cor:diffstat}\label{sec:strongsep}}
One can strengthen Corollary \ref{cor:diffstat} so that the separation between the success probabilities
of the optimal acyclic network and the optimal static assignment is $1-2^{-\Omega(n)}$ (instead of $1-2^{-\Omega(\sqrt{n})}$).  
The construction of the objective function in the following proof was suggested by Madhu Sudan.

\begin{corollary}
There is a robustly expressible family of functions $f_n : \pcube^n \to \pcube$ that cannot be
robustly expressed by a static assignment.  For any constant $\eps > 0$ and for any static
assignment of $\{x_1,\dots,x_n\}$, the probability that an $\eps$-mutation of the assignment satisfies $f_n$
is at most $2^{-\Omega(n)}$, while there is a $(\eps,1-2^{-\Omega(n)})$-robust acyclic network with respect
to $f_n$.  
\end{corollary}
\begin{proof}
For each $n \geq 1$, consider the function $f_n : \pcube^n \to \pcube$ where $f_n(x_1,\dots,x_n) =
\sgn(x_1 + x_1x_2 + x_1x_2x_3 + \cdots + x_1x_2\cdots x_n - \frac{n}{4})$.  Noting that the
term-variable graph of $p(x) \defeq x_1 + x_1x_2 + x_1x_2x_3 + \cdots + x_1x_2\cdots x_n$ is sequentially covered by the sequence 
$x_n,\dots,x_1$, it follows by a probabilistic argument similar to the one in the proof of Theorem
\ref{thm:main-neces}, that the function $f_n$ is $\eps$-robustly expressible for a small enough
constant $\eps$. 

On the other hand, we next show that $f_n$ cannot be robustly expressed by any static assignment.  Fix a static
assignment for $f_n$, and consider an $\eps$-mutation of it.   Then, each $x_i$ is an independent
random variable that acquires $-1/1$ with probability $1-\eps$ and $1/-1$ with probability $\eps$.
For $i \in \{1,\dots,n\}$, let $y_i = x_1x_2\cdots x_i$.  Now, $p(x) = \sum_i y_i$, and therefore,
$|\Exp{}{[p(x)]}| = \sum_i |\Exp{}{y_i}| \leq \sum_i (1-2\eps)^i \leq \frac{1-2\eps}{2\eps}$, a constant.  We
need to bound the concentration around this mean.  Note that the $y_i$'s are not independent; instead, their generation process is exactly captured by the well studied memoryless Markov process. (for an introduction to Markov processes and eigenvalues see \cite{meyn1993mca}) The Markov process on hand is characterized by $\Pr[y_i = a_i | y_{i-1}=a_{i-1}]$ and is specified by a 2-by-2 stochastic matrix, either $\left( \begin{array}{cc}
1-\eps & \eps \\
\eps & 1-\eps  \end{array} \right)$ or $\left( \begin{array}{cc}
\eps & 1-\eps \\
1-\eps & \eps  \end{array} \right)$.  The eigenvalue gaps of these two matrices are $2\eps$
and $2(1-\eps)$ respectively.  By a concentration bound on the sum of elements generated by a Markov
chain with eigenvalue gap $\delta$, given in Theorem 4.23 of \cite{DubhashiPanconesi}, we have that
$\Pr\left[|\Exp{}{\left[\sum_iy_i\right]}- \sum_i y_i| > n/8\right] \leq 2^{-\Omega(\delta n)}$.  So, with probability
at least $1-2^{-\Omega(n)}$, $\sum_i y_i < n/4$ and $f_n$ is not satisfied.
\end{proof}

\end{document}